\documentclass[runningheads]{llncs}
\usepackage[T1]{fontenc}
\usepackage{graphicx}
\usepackage{amsmath, amssymb}
\usepackage{mathtools}
\usepackage{booktabs}
\usepackage{algorithm}
\usepackage{algpseudocode}

\usepackage{enumerate}

\newcommand{\precedes}{\textsc{Precedes}}
\newcommand{\concurrent}{\textsc{Concurrent}}
\newcommand{\addinv}{\textsc{Add\_Inv}}
\newcommand{\addres}{\textsc{Add\_Res}}

\newcommand{\lp}{\mathit{lp}}

\newcommand{\Reg}{\mathit{Reg}}
\newcommand{\inv}{\mathit{inv}}
\newcommand{\res}{\mathit{res}}

\newcommand{\hbL}[1][]{\prec^{L}_{#1}}

\newcommand{\Eall}  {E_{\mathit{all}}}
\newcommand{\Eobs}  {E_{\mathit{obs}}}

\newcommand{\einv}{\widehat{\mathrm{inv}}}
\newcommand{\eres}{\widehat{\mathrm{res}}}

\begin{document}

\title{On the Limits of Causal Observation in Shared-Memory Systems \thanks{This is the full version of a paper accepted for publication
at the 28th International Symposium on Stabilization, Safety, and
Security of Distributed Systems (SSS 2026).}}

\author{Gilde Valeria Rodríguez\inst{1} \and
Armando Castañeda\inst{2} \and
Miguel Piña\inst{3}}

\authorrunning{G. V. Rodríguez, A. Castañeda, M. Piña}

\institute{Posgrado en Ciencia e Ingeniería de la Computación, Universidad Nacional Autónoma de México, Mexico \and
Instituto de Matemáticas, Universidad Nacional Autónoma de México, Mexico \and
Independent Researcher, Mexico}

\maketitle              


\begin{abstract}
Reasoning about concurrent system correctness requires detecting real-time operation precedence ($<_H$). In asynchronous shared-memory systems, real-time precedence and information-flow causality diverge, rendering causality-tracking techniques like vector clocks incapable of recovering $<_H$. An external monitor observing only operation boundaries faces an inherent limitation: it must either miss genuine precedences or report spurious ones. We formalize this challenge as the Causal Observability Problem ($\mathbf{COP}$) and prove that strong consistency---achieving both soundness and completeness simultaneously---is unachievable at the observable boundary. We then establish a tight placement dichotomy: instrumenting externally guarantees soundness, whereas instrumenting internally guarantees completeness, independently of the timestamping algorithm. To construct sound monitors, we analyze three non-blocking implementations: \textsc{FAInc} (a centralized counter), \textsc{Striped} (a decentralized counter), and \textsc{Collect} (an iterative process snapshot). We prove that both \textsc{FAInc} and \textsc{Striped} are linearizable, whereas \textsc{Collect} fails all standard object consistency models---satisfying neither linearizability, sequential consistency, nor quiescent consistency. Crucially, we prove an invariance result: all three implementations provide identical $\mathbf{COP}$ soundness guarantees under external placement. This demonstrates a fundamental theoretical decoupling---observational soundness operates on pairwise interval precedence and does not require the global sequential witness that standard object consistency demands. Empirical evaluation on a 64-core NUMA system shows that \textsc{Striped} matches the high throughput and scalability of \textsc{Collect} while preserving linearizability and eliminating central counter contention.
\end{abstract}


\section{Introduction}
\label{sec:intro}

Reasoning about the correctness of a concurrent system requires
knowing, for any two operations, whether one causally preceded the
other. Two distinct notions answer this question. The first is
\emph{real-time precedence}: operation $A$ precedes $B$ if $A$'s
response occurs, in physical time, before $B$'s invocation---denoted $<_H$, this is what linearizability requires an
implementation to respect~\cite{HW90}. The second is
\emph{information-flow causality}: $A$ precedes $B$ if some chain of
shared-memory reads and writes carries information from $A$ to $B$---Lamport's happened-before relation~\cite{Lamport78}, exactly
reconstructed by vector clocks~\cite{Fidge88,Mattern88}. In an
asynchronous system, where no process has access to a physical
clock, these two notions diverge: $A$ may complete before $B$ begins
with no information ever exchanged between them, and concurrent
operations may exchange information without either preceding the
other in real time. Consequently, no clock-based or causality-tracking
technique can recover $<_H$ from information flow alone.

This divergence is exacerbated by a fundamental limitation: an external \emph{monitor}---a decoupled algorithm distinct from the target system it observes---sees only operation \emph{boundaries} (invocations and responses), rather than the internal memory accesses that determine $<_H$. This forms the core obstacle in \emph{runtime verification of linearizability}~\cite{CR23,RPC26}, where it was established~\cite{CR23,CR26_JACM} that no monitor restricted to observable boundaries can achieve both sound and complete verification. That proof rests
on an implicit, unstated obstacle: no such monitor can reliably
detect real-time precedence $<_H$ between two operations from their
boundary events alone---it must either miss genuine precedences or
report spurious ones. This paper makes that obstacle explicit and
studies it directly, independently of linearizability verification.
We formalize it as the \emph{Causal Observability Problem}
($\mathbf{COP}$): assign timestamps to boundary events that
faithfully reflect $<_H$, where \emph{completeness} means capturing
every genuine precedence and \emph{soundness} means reporting only
genuine ones. We prove rigorously that \emph{strong consistency}---both simultaneously---is unachievable at the observable boundary,
and then ask what \emph{is} achievable.

Instead, the specific guarantees that \emph{can} be achieved depend 
on a single structural decision: where the monitor places its 
instrumentation relative to the true operation boundaries. Studying the 
strategies used by existing monitors~\cite{CR23,CR26_JACM,RC24,EF18}, we 
identify a tight placement dichotomy. Instrumenting \emph{externally}---writing 
events \emph{before} invoking and \emph{after} completing each operation, 
enclosing the true interval---yields \emph{soundness}. Instrumenting 
\emph{internally}---\emph{after} the operation starts and \emph{before} 
it finishes---yields \emph{completeness}. Instrumenting only one boundary 
yields neither~\cite{EF18}. No placement achieves both, as we prove via 
an indistinguishability argument holding regardless of the shared-memory 
objects used.

To implement sound monitors without blocking bottlenecks, we analyze
three non-blocking implementations. \textsc{FAInc}, a centralized
atomic counter introduced in~\cite{RC24}, is linearizable.
\textsc{Striped}, a decentralized counter built on Java's
\texttt{LongAdder} restricted to increment-only operations and
introduced in this work, is also linearizable---the restriction to
increments is essential, as it preserves the monotonicity argument
linearizability depends on. \textsc{Collect}, which iteratively
snapshots process registers~\cite{RC24}, turns out to admit no legal
sequential witness at all: it satisfies none of linearizability,
sequential consistency, or quiescent consistency. Neither this fact
nor \textsc{FAInc}'s consistency was established in~\cite{RC24}. A
key result of this paper is that, despite this gap, all three
implementations provide identical $\mathbf{COP}$ guarantees:
placement alone determines observable behavior, and $\mathbf{COP}$-soundness
never requires the sequential witness that classical consistency
conditions do. We validate this empirically on a 64-core NUMA
architecture: prior work~\cite{RPC26} showed \textsc{Collect}
outperforms \textsc{FAInc} in scalability, and we show
\textsc{Striped} matches that throughput while remaining linearizable.

\begin{table}[t]
\centering
\footnotesize
\setlength{\tabcolsep}{3pt}
\renewcommand{\arraystretch}{1.05}
\caption{Prior work vs.\ contributions of this paper.}
\label{tab:contributions}
\begin{tabular}{p{0.14\textwidth} p{0.33\textwidth} p{0.48\textwidth}}
\toprule
\textbf{Aspect} & \textbf{Prior work} & \textbf{This paper} \\
\midrule
Boundary\newline impossibility & Limitation in verification of \textit{lin} at runtime (RVL)~\cite{CR23,CR26_JACM} & Formalized as $\mathbf{COP}$; proved impossible under ideal log (Thm.~\ref{thm:impossibility}) \\
\hline
Placement\newline dichotomy & Ad hoc placement in existing monitors~\cite{CR23,RC24,EF18} & Tight dichotomy: external $\Rightarrow$ sound, internal $\Rightarrow$ complete (Lemma~\ref{lem:placement}) \\
\hline
FAInc \& Collect & Introduced without consistency analysis~\cite{RC24} & FAInc is lin; Collect satisfies \emph{no} consistency model (L, SC, QC) (Thm.~\ref{thm:collect-no-witness}) \\
\hline
\textsc{Striped} & LongAdder design in Java~\cite{lea-longadder} & Adapted as linearizable counter for Causal Monitor; eliminates contention (\S\ref{sec:counter}) \\
\hline
$\mathbf{COP}$\newline invariance & Proved both solve RVL~\cite{RC24}, treating them as equivalent & Proved $\mathbf{COP}$-equivalent under ext.\ placement: soundness needs no seq.\ witness (\S\ref{sec:correctness}) \\
\hline
Performance\newline (64 cores) & FAInc suffers cache bouncing (tail $>1200$\,ms)~\cite{RPC26} & \textsc{Striped} matches Collect ($\approx 1.2\times$ faster, $2\times$ lower tail latency) while lin (\S\ref{sec:evaluation}) \\
\bottomrule
\end{tabular}
\end{table}

\paragraph{Contributions.}
We summarize our main contributions below and contrast them with prior work in Table~\ref{tab:contributions}:
\begin{enumerate}[(i)]
  \setlength{\itemsep}{1.5pt}
  \setlength{\parskip}{0pt}
  \setlength{\parsep}{0pt}
  \setlength{\topsep}{2pt}
\item \textbf{Formalization of the observability limit.} We define $\mathbf{COP}$ and prove that strong consistency is unachievable at the observable boundary, extracting the root obstacle behind the impossibility of sound and complete runtime verification of linearizability~\cite{CR23,CR26_JACM}.
\item \textbf{Placement dichotomy.} We prove that instrumentation placement---internal or external---determines whether a monitor achieves soundness or completeness, independently of the underlying timestamp algorithm.
  \item \textbf{Consistency analysis of the Causal Monitor.} We establish that the Counter Monitor is linearizable while the Collect Monitor admits no legal sequential witness at all (failing L, SC, and QC).
  \item \textbf{$\mathbf{COP}$ invariance.} We prove that all three implementations provide identical $\mathbf{COP}$ guarantees, showing that observational soundness is independent of internal object consistency.
  \item[(v)] \textbf{Scalable linearizable implementation.}
  We introduce \textsc{Striped} and show empirically that it matches
  \textsc{Collect}'s scalability while eliminating \textsc{FAInc}'s
  contention bottleneck, achieving the advantage prior
  work~\cite{RPC26} showed only for \textsc{Collect}.
\end{enumerate}

Section~\ref{sec:model} presents the system model.
Section~\ref{sec:cop} defines $\mathbf{COP}$ and proves the
impossibility of strong consistency.
Section~\ref{sec:eint} analyzes placement strategies and proves the
completeness/soundness characterization.
Section~\ref{sec:monitor} specifies the Causal Monitor object.
Section~\ref{sec:implementations} presents the Counter and Collect
implementations.
Section~\ref{sec:correctness} contains the correctness proofs and the
$\mathbf{COP}$ invariance result.
Section~\ref{sec:evaluation} presents the empirical evaluation.
Section~\ref{sec:discussion} discusses related work.

Due to lack of space, some proofs and additional discussions are
omitted. The full version of the paper can be found
in~\cite{completeVersionRCP26}.

\section{Model}
\label{sec:model}

We consider $N$ asynchronous processes $p_1, \dots, p_N$ communicating via atomic shared memory~\cite{HW90}. Processes interact with typed shared objects by issuing an \emph{invocation} $\mathit{inv}(\mathit{op})$ and receiving a \emph{response} $\mathit{res}(\mathit{op})$~\cite{HW90}. An operation is \emph{pending} if invoked but not completed. Events occur instantaneously at a single real-time instant ($e \prec f$).
A \emph{history} $H$ is a finite sequence of invocation and response events. $H$ is \emph{well-formed} if every response is preceded by a matching invocation; \emph{sequential} if invocations are immediately followed by their matching responses; and \emph{complete} if no operation is pending.

\begin{definition}[Completion \& Equivalence]
\label{def:equiv}
A \emph{completion} $\overline{H}$ of a well-formed history $H$ appends matching responses to or drops pending operations. Two complete, well-formed histories $H, H'$ are \emph{equivalent} ($H \equiv H'$) if a bijection $\pi$ maps operations of $H$ to operations of $H'$ preserving operation type, arguments, and return values.
\end{definition}

The \emph{real-time order} $<_H$ on operations in $H$ is defined by $\mathit{op}_1 <_H \mathit{op}_2 \iff \mathit{res}(\mathit{op}_1) \prec \mathit{inv}(\mathit{op}_2)$. Operations are \emph{concurrent} ($\mathit{op}_1 \parallel \mathit{op}_2$) if neither precedes the other. An object type $T$ is defined by a \emph{sequential specification} $\mathit{Seq}(T)$, a prefix-closed set of legal sequential histories~\cite{HW90}.

\begin{definition}[Linearizability~{\cite{HW90}}]
\label{def:lin}
An implementation of type $T$ is \emph{linearizable} w.r.t.\ $\mathit{Seq}(T)$ if for every well-formed history $H$ there exists a legal sequential history $S \in \mathit{Seq}(T)$ such that: \textbf{(L1)} $S \equiv \overline{H}$ for some completion $\overline{H}$ of $H$ (Definition~\ref{def:equiv}), and \textbf{(L2)} $\mathit{op}_1 <_H \mathit{op}_2 \implies \mathit{op}_1 <_S \mathit{op}_2$.
\end{definition}

\begin{definition}[Sequential Consistency~{\cite{Lamport97}}]
\label{def:sc}
An implementation is \emph{sequentially consistent} w.r.t.\ $\mathit{Seq}(T)$ if for every well-formed $H$ there exists a legal $S \in \mathit{Seq}(T)$ such that: \textbf{(SC1)} $S \equiv \overline{H}$ for some completion $\overline{H}$, and \textbf{(SC2)} $S$ respects program order (for each process $p_i$, operations appear in $S$ in the same order as in $H$).
\end{definition}

\begin{definition}[Quiescence \& Quiescent Consistency~{\cite{HS08}}]
\label{def:qc}
An instant $Q$ is \emph{quiescent} in $H$ if no operation is pending at $Q$ ($\mathit{inv}(\mathit{op}) \prec Q \implies \mathit{res}(\mathit{op}) \prec Q$). Operations not separated by a quiescent instant belong to the same \emph{quiescent group}. An implementation is \emph{quiescently consistent} w.r.t.\ $\mathit{Seq}(T)$ if for every well-formed $H$ there exists a legal $S \in \mathit{Seq}(T)$ such that: \textbf{(QC1)} $S \equiv \overline{H}$ for some completion $\overline{H}$, and \textbf{(QC2)} $\mathit{res}(\mathit{op}_1) \prec Q \prec \mathit{inv}(\mathit{op}_2) \implies \mathit{op}_1 <_S \mathit{op}_2$ for any quiescent instant $Q$.
\end{definition}

Linearizability is strictly stronger than both SC and QC ($\text{Lin} \subsetneq \text{SC}$ and $\text{Lin} \subsetneq \text{QC}$). Sequential consistency and quiescent consistency are \emph{incomparable}: SC permits reordering operations separated by $Q$ across different processes (violating QC2), whereas QC permits reordering same-process operations within the same quiescent group (violating SC2).

\section{The Causal Observability Problem}
\label{sec:cop}

In a concurrent execution, an external client observes only the
boundaries of operations---their invocations and responses---and uses this information to reason about the order in which
operations executed.
The client's goal is to determine the real-time order $<_H$:
whether one operation completed before another began.

A \emph{timestamping protocol} runs inside the system, on the
processes themselves, and assigns timestamps to boundary events
so that the client can answer order queries by comparing
timestamps.
The protocol has access to the flow of information through shared
memory---it can observe Lamport causality $\hbL$---but it
cannot directly observe the true operation boundaries, which are
transfers of control rather than shared-memory steps.
The central question is whether a protocol can assign timestamps
that faithfully reflect $<_H$ using only what it can observe.

\subsection{Observable Events and Lamport Causality}
\label{sec:obs-events}

Recall from Section~\ref{sec:model} that an execution produces a
history $H$: a sequence of invocation and response events over the
processes $p_1,\dots,p_N$.
Besides these boundary events, an execution also contains
\emph{internal} events: the shared-memory steps (reads and writes)
that processes perform between an invocation and its matching
response while executing the target algorithm. We write
$\mathit{Events}$ for the set of \emph{all} events of an execution,
boundary and internal alike, and reserve $H$ for the subsequence of
$\mathit{Events}$ consisting only of invocation and response events,
in the order they occur.

The \emph{observable boundary} is the set of boundary events that
belong to $H$:
$\Eobs \;\triangleq\;
  \bigl\{\,\inv(\mathit{op}),\;\res(\mathit{op})
  \;\bigm|\; \mathit{op} \in H \,\bigr\} \;\subseteq\; \mathit{Events}$.
While $H$ is a \emph{sequence}---it records the real-time order in
which invocations and responses occur---$\Eobs$ is the underlying
\emph{set} of those same events, stripped of order, and augmented
below with the causal relation $\hbL$ restricted to it. $\Eobs$ is
what a monitor may instrument (Section~\ref{sec:eint}) and to
which a timestamping protocol assigns timestamps; $H$ remains the
object whose real-time order $<_H$ the protocol must faithfully
report (Definition~\ref{def:cop} below).

The order the protocol \emph{must report} is the real-time order
$<_H$ of Section~\ref{sec:model}.
Since $\mathit{op}_A <_H \mathit{op}_B$ holds exactly when
$\res(\mathit{op}_A) \prec \inv(\mathit{op}_B)$, this is already
a relation on boundary events: it links the response of one
operation to the invocation of a later one.

The order the protocol \emph{can construct} from shared memory
is Lamport's happened-before~\cite{Lamport78}.

\begin{definition}[Lamport causality, $\hbL$]
\label{def:lamport}
$\hbL$ is the smallest transitive relation over $\mathit{Events}$
(boundary and internal events alike) such that:
\begin{enumerate}
  \item \emph{Program order}: if $e_1$ and $e_2$ occur at the same
        process and $e_1 \prec e_2$, then $e_1 \hbL e_2$.
  \item \emph{Object communication}: if $e_1$ is a \emph{shared-memory
        step} on a shared object $O$ whose effect is observable by
        $e_2$---that is, the state change produced by $e_1$ influences
        the response of $e_2$---then $e_1 \hbL e_2$.
  \item \emph{Transitivity}: if $e_1 \hbL e_2$ and $e_2 \hbL e_3$,
        then $e_1 \hbL e_3$.
\end{enumerate}
Clauses~1 and~3 range over arbitrary events of $\mathit{Events}$,
including internal shared-memory steps; clause~2 specifically
concerns internal shared-memory steps (reads and writes on shared
objects), as opposed to the control-transfer boundary events of
$\Eobs$, which are never themselves shared-memory steps
(Section~\ref{sec:eint}).
\end{definition}

\noindent
In the special case where shared objects are read/write registers,
condition~(2) reduces to the standard read-from relation:
$e_1$ writes a value to register $r$ and $e_2$ reads that value
from $r$.

\noindent
Restricted to $\Eobs$, we write $\hbL[\Eobs]$ for the relation
induced by $\hbL$ on boundary events: $e_1 \hbL[\Eobs] e_2$ for
$e_1, e_2 \in \Eobs$ iff $e_1 \hbL e_2$, where the underlying
$\hbL$-chain may pass through internal events not in $\Eobs$.
Vector clocks~\cite{Fidge88,Mattern88} reconstruct $\hbL[\Eobs]$
exactly: $e_1 \hbL[\Eobs] e_2 \iff \mathit{vc}(e_1) < \mathit{vc}(e_2)$.
However, as we show below, $\hbL[\Eobs]$ is not the relation the
protocol needs to report.

\subsection{The Problem}
\label{sec:cop-def}

Let $(\mathbb{T}, <_\mathbb{T})$ be a totally ordered set.
A \emph{timestamping protocol} assigns $\ts(e) \in \mathbb{T}$
to each $e \in \Eobs$.
The client then answers a precedence query by checking whether
$\ts(\res(\mathit{op}_A)) <_\mathbb{T} \ts(\inv(\mathit{op}_B))$.

\begin{definition}[Causal Observability Problem, $\mathbf{COP}$]
\label{def:cop}
The \emph{Causal Observability Problem} asks for a timestamping
protocol assigning $\ts(e) \in \mathbb{T}$ to each $e \in \Eobs$
such that, for all operations $\mathit{op}_A, \mathit{op}_B$:
\begin{enumerate}
  \item \emph{Completeness}:
        $\mathit{op}_A <_H \mathit{op}_B \;\Rightarrow\; \ts(\res(\mathit{op}_A)) <_\mathbb{T} \ts(\inv(\mathit{op}_B))$.
  \item \emph{Soundness}:
        $\ts(\res(\mathit{op}_A)) <_\mathbb{T} \ts(\inv(\mathit{op}_B)) \;\Rightarrow\; \mathit{op}_A <_H \mathit{op}_B$.
\end{enumerate}
A protocol satisfying only~(1) is \emph{complete};
one satisfying only~(2) is \emph{sound};
one satisfying both is \emph{strongly consistent}.
\end{definition}

\noindent
Both conditions are stated directly in terms of $<_H$: the
protocol's timestamps must faithfully reflect real-time operation
order, not merely information flow.

Completeness forbids \emph{missing} a genuine precedence: if
$\mathit{op}_A$ truly completed before $\mathit{op}_B$ began,
the timestamps must reflect that.
Soundness forbids \emph{reporting} a spurious precedence: if the
timestamps say $\mathit{op}_A$ preceded $\mathit{op}_B$, then
$\mathit{op}_A$ must have truly completed before $\mathit{op}_B$
began.
A strongly consistent protocol does both: its timestamp order
is equivalent to $<_H$.

Note that soundness does \emph{not} require timestamps to imply
precedence for concurrent operations --- two concurrent operations
will inevitably be ordered by any total order on $\mathbb{T}$,
and that is acceptable.
What soundness prohibits is asserting a precedence that did not
occur in real time.

\subsection{The Gap Between $<_H$ and $\hbL[\Eobs]$}
\label{sec:gap}

The order the protocol must report, $<_H$, and the order it can
construct, $\hbL[\Eobs]$, both live on $\Eobs$.
They are, however, incomparable as relations.

\begin{claim}
\label{claim:incomparable}
$<_H \,\not\subseteq\, \hbL[\Eobs]$ and
$\hbL[\Eobs] \,\not\subseteq\, {<_H}$.
\end{claim}

\begin{proof}
\emph{$<_H \not\subseteq \hbL[\Eobs]$.}
Let $\mathit{op}_A$ complete at $p_i$ before $\mathit{op}_B$
begins at $p_j$, with no register written by $p_i$ during
$\mathit{op}_A$ subsequently read by $p_j$.
Then $\res(\mathit{op}_A) \prec \inv(\mathit{op}_B)$, so
$\mathit{op}_A <_H \mathit{op}_B$, yet no $\hbL$-chain links
$\res(\mathit{op}_A)$ to $\inv(\mathit{op}_B)$: time elapsed,
but no information was exchanged.

\emph{$\hbL[\Eobs] \not\subseteq {<_H}$.}
Let $\mathit{op}_A$ and $\mathit{op}_B$ overlap
($\mathit{op}_A \parallel \mathit{op}_B$), and let $p_i$ write a
register during $\mathit{op}_A$ that $p_j$ reads during
$\mathit{op}_B$.
By clauses~1 and~2 of Definition~\ref{def:lamport},
$\inv(\mathit{op}_A) \hbL[\Eobs] \res(\mathit{op}_B)$,
yet $\mathit{op}_A \parallel \mathit{op}_B$ means neither
$\mathit{op}_A <_H \mathit{op}_B$ nor
$\mathit{op}_B <_H \mathit{op}_A$.
\qed
\end{proof}

\noindent
The two failures have distinct causes.
$<_H$ sees real-time precedences that leave no informational
trace: one operation finishes before another begins, but the
processes never communicate.
$\hbL[\Eobs]$ sees informational chains that carry no real-time
precedence: concurrent operations can exchange information through
shared registers without one completing before the other starts.

This incomparability explains why vector clocks do not solve
$\mathbf{COP}$: they reconstruct $\hbL[\Eobs]$ exactly, but
that is not $<_H$.
A protocol that tracks $\hbL[\Eobs]$ will miss real-time
precedences that left no informational trace, violating
completeness.
A protocol that reports every $<_H$ precedence may assign
timestamps that contradict $\hbL[\Eobs]$, violating soundness.
The following section shows that no protocol can be strongly
consistent for $\mathbf{COP}$: completeness and soundness are
irreconcilable at $\Eobs$.

\section{The Interface Event Layer}
\label{sec:eint}
Section~\ref{sec:cop} showed that the order a protocol must report
($<_H$) and the order it can construct from shared memory
($\hbL[\Eobs]$) are incomparable. A protocol approximating $<_H$
must therefore construct $\hbL[\Eobs]$ by invoking shared-memory
objects around the boundaries it observes, using the causal
relationships among those invocations as a proxy for real-time
order. Studying the strategies used by existing
monitors~\cite{CR23,CR26_JACM,RC24,EF18}, we identify two classes.
Monitors invoking objects at only one boundary are neither sound nor
complete~\cite{EF18} (Lemma~\ref{lem:unilateral}). Monitors invoking objects at both
boundaries split into those whose invocations fall strictly inside
the operation interval (\emph{internal placement}), yielding
completeness, and those that strictly enclose it (\emph{external
placement}), yielding soundness
(Lemmas~\ref{lem:placement}--\ref{lem:internal-not-sound}); no combination achieves both
(Theorem~\ref{thm:impossibility}).

\subsection{Interface Events and Monitors}

A protocol that monitors $\mathbf{COP}$ cannot observe the true
boundaries $\inv(\mathit{op})$ and $\res(\mathit{op})$ directly:
these mark the transfer of control between a process and its
caller, and are not steps in shared memory.
The only tool available is to invoke shared-memory objects around
those boundaries, producing observable actions whose causal
relationships reconstruct $\hbL[\Eobs]$.

\begin{definition}[Monitor]
\label{def:monitor}
A \emph{monitor} for $\mathbf{COP}$ is a protocol that, for each
operation $\mathit{op}$, invokes a finite set of shared-memory
objects before and after each true boundary, producing observable
actions called \emph{interface events}.
The monitor answers $\precedes(A, B) = \mathtt{true}$ iff the
interface events establish $\eres(\mathit{op}_A) \hbL \einv(\mathit{op}_B)$.
\end{definition}

\noindent
We place no restriction on what objects the monitor may invoke:
interface events may be operations on objects of any type,
including objects with infinite consensus number~\cite{H91}.
The obstacle is not the synchronization power of the objects
but the position of their invocations relative to the true
boundaries.

We assume every shared-memory operation takes effect at a single
linearization point, formalized through an \emph{ideal log}.

\begin{definition}[Ideal log]
\label{def:ideal-log}
A monitor operates under an \emph{ideal log} if every interface
event $e$ has a linearization point $\lp(e)$ such that: $e_1 \hbL e_2 \;\iff\; \lp(e_1) \prec \lp(e_2)$.
\end{definition}

\noindent
An ideal log grants the monitor perfect knowledge of the
real-time order among its own interface events.
Any impossibility result under an ideal log is unconditional:
it cannot be attributed to imprecision in logging or to the
choice of objects invoked.

Because what matters for a precedence query is the last interface
event before $\res(\mathit{op})$ and the first after
$\inv(\mathit{op})$, intermediate events carry no additional
causal information.
Without loss of generality, we model each monitor as emitting
exactly one interface event per boundary: $\einv(\mathit{op})$
and $\eres(\mathit{op})$.

\subsection{The Placement Dichotomy}

Under an ideal log, each interface event and each true boundary
occur at distinct points in real time, since interface events
are shared-memory invocations and true boundaries are control
transfers.
This forces every monitor that instruments both boundaries into
one of two classes.

We first dispose of monitors that instrument only one boundary.

\begin{lemma}[Unilateral placement achieves neither soundness nor
completeness]
\label{lem:unilateral}
Let $\mathcal{M}$ be a monitor that invokes objects only before
both boundaries (only before $\inv(\mathit{op})$ and
$\res(\mathit{op})$) or only after both boundaries.
Then $\mathcal{M}$ is neither sound nor complete for
$\mathbf{COP}$.
\end{lemma}

\begin{proof}
We prove both cases.

\emph{Case 1: only before both boundaries.}
The monitor invokes objects before $\inv(\mathit{op})$ and before
$\res(\mathit{op})$, producing interface events
$\einv(\mathit{op})$ and $\eres(\mathit{op})$ with
$\lp(\einv(\mathit{op})) \prec \inv(\mathit{op})$ and
$\lp(\eres(\mathit{op})) \prec \res(\mathit{op})$.

\emph{Not complete.}
Consider $\mathit{op}_A <_H \mathit{op}_B$, so
$\res(\mathit{op}_A) \prec \inv(\mathit{op}_B)$.
The monitor's event for the response of $\mathit{op}_A$ fires
before $\res(\mathit{op}_A)$, and the monitor's event for the
invocation of $\mathit{op}_B$ fires before $\inv(\mathit{op}_B)$.
No causal chain is established between $\eres(\mathit{op}_A)$ and
$\einv(\mathit{op}_B)$: $\eres(\mathit{op}_A)$ fires before
$\res(\mathit{op}_A)$ and $\einv(\mathit{op}_B)$ fires before
$\inv(\mathit{op}_B)$, but $\eres(\mathit{op}_A)$ may fire after
$\einv(\mathit{op}_B)$ even when $\mathit{op}_A <_H \mathit{op}_B$.
Concretely: if $p_i$ delays firing $\eres(\mathit{op}_A)$ until
just before $\res(\mathit{op}_A)$, and $p_j$ fires
$\einv(\mathit{op}_B)$ early, the real-time order of the interface
events may be $\einv(\mathit{op}_B) \prec \eres(\mathit{op}_A)$
despite $\mathit{op}_A <_H \mathit{op}_B$.
The monitor misses a genuine precedence.

\emph{Not sound.}
Consider $\mathit{op}_A \parallel \mathit{op}_B$.
The monitor fires $\eres(\mathit{op}_A)$ before $\res(\mathit{op}_A)$
and $\einv(\mathit{op}_B)$ before $\inv(\mathit{op}_B)$.
If $\eres(\mathit{op}_A)$ happens to fire before
$\einv(\mathit{op}_B)$ in real time, the monitor reports
$\precedes(A,B) = \mathtt{true}$, but $\mathit{op}_A \parallel
\mathit{op}_B$ means $\mathit{op}_A <_H \mathit{op}_B$ is false.
A spurious precedence is reported.

\emph{Case 2: only after both boundaries.}
The monitor invokes objects after $\inv(\mathit{op})$ and after
$\res(\mathit{op})$, producing interface events with
$\inv(\mathit{op}) \prec \lp(\einv(\mathit{op}))$ and
$\res(\mathit{op}) \prec \lp(\eres(\mathit{op}))$.

\emph{Not complete.}
Consider $\mathit{op}_A <_H \mathit{op}_B$, so
$\res(\mathit{op}_A) \prec \inv(\mathit{op}_B)$.
The monitor fires $\eres(\mathit{op}_A)$ after $\res(\mathit{op}_A)$
and $\einv(\mathit{op}_B)$ after $\inv(\mathit{op}_B)$.
If $p_j$ fires $\einv(\mathit{op}_B)$ immediately after
$\inv(\mathit{op}_B)$ while $p_i$ delays $\eres(\mathit{op}_A)$,
the real-time order may be $\einv(\mathit{op}_B) \prec
\eres(\mathit{op}_A)$ despite $\mathit{op}_A <_H \mathit{op}_B$.
The monitor misses a genuine precedence.

\emph{Not sound.}
Consider $\mathit{op}_A \parallel \mathit{op}_B$.
The monitor fires $\eres(\mathit{op}_A)$ after $\res(\mathit{op}_A)$
and $\einv(\mathit{op}_B)$ after $\inv(\mathit{op}_B)$.
If $\eres(\mathit{op}_A)$ fires before $\einv(\mathit{op}_B)$,
the monitor reports $\precedes(A,B) = \mathtt{true}$, but
$\mathit{op}_A \parallel \mathit{op}_B$ so no real precedence
exists.
A spurious precedence is reported.
\qed
\end{proof}

\begin{definition}[Placement]
\label{def:placement}
A monitor has \emph{internal placement} if
$\inv(\mathit{op}) \prec \lp(\einv(\mathit{op}))$ and
$\lp(\eres(\mathit{op})) \prec \res(\mathit{op})$, and
\emph{external placement} if
$\lp(\einv(\mathit{op})) \prec \inv(\mathit{op})$ and
$\res(\mathit{op}) \prec \lp(\eres(\mathit{op}))$.
\end{definition}

\noindent
Internal placement corresponds to online checkers that
instrument the object's internal steps~\cite{CR23}.
External placement corresponds to offline verifiers that wrap
each call site~\cite{CR26_JACM}.

\begin{lemma}[Placement characterization]
\label{lem:placement}
Assuming an ideal log, a monitor with internal placement is
complete for $\mathbf{COP}$, and a monitor with external
placement is sound.
Neither achieves the other property.
\end{lemma}

\begin{proof}
\emph{External placement is sound.} By Definition~\ref{def:placement},
external placement gives $\lp(\einv(\mathit{op})) \prec \inv(\mathit{op})$
and $\res(\mathit{op}) \prec \lp(\eres(\mathit{op}))$ for every
operation. Suppose $\precedes(A,B) = \mathtt{true}$: by
Definition~\ref{def:monitor} and the ideal log
(Definition~\ref{def:ideal-log}), this means
$\lp(\eres(\mathit{op}_A)) \prec \lp(\einv(\mathit{op}_B))$.
Chaining, $\res(\mathit{op}_A) \prec \lp(\eres(\mathit{op}_A)) \prec
\lp(\einv(\mathit{op}_B)) \prec \inv(\mathit{op}_B)$, so
$\res(\mathit{op}_A) \prec \inv(\mathit{op}_B)$, i.e.\
$\mathit{op}_A <_H \mathit{op}_B$. Soundness holds.

\emph{Internal placement is complete.} By
Definition~\ref{def:placement}, internal placement gives
$\inv(\mathit{op}) \prec \lp(\einv(\mathit{op}))$ and
$\lp(\eres(\mathit{op})) \prec \res(\mathit{op})$ for every
operation. Suppose $\mathit{op}_A <_H \mathit{op}_B$, i.e.\
$\res(\mathit{op}_A) \prec \inv(\mathit{op}_B)$. Chaining,
$\lp(\eres(\mathit{op}_A)) \prec \res(\mathit{op}_A) \prec
\inv(\mathit{op}_B) \prec \lp(\einv(\mathit{op}_B))$, so
$\lp(\eres(\mathit{op}_A)) \prec \lp(\einv(\mathit{op}_B))$, hence
by the ideal log $\eres(\mathit{op}_A) \hbL \einv(\mathit{op}_B)$
and the monitor reports $\precedes(A,B) = \mathtt{true}$.
Completeness holds.

\emph{Neither achieves the other property} is established by
Lemmas~\ref{lem:external-not-complete}
and~\ref{lem:internal-not-sound} below.
\qed
\end{proof}

\begin{lemma}[External placement is not complete]
\label{lem:external-not-complete}
A monitor with external placement is sound for $\mathbf{COP}$
(Lemma~\ref{lem:placement}) but not complete: there exist
executions where $\mathit{op}_A <_H \mathit{op}_B$ yet
$\precedes(A,B) = \mathtt{false}$.
\end{lemma}

\begin{proof}
External placement gives
$\lp(\einv(\mathit{op})) \prec \inv(\mathit{op})$ and
$\res(\mathit{op}) \prec \lp(\eres(\mathit{op}))$.

Consider two operations $\mathit{op}_A$ on $p_i$ and
$\mathit{op}_B$ on $p_j$ with $\mathit{op}_A <_H \mathit{op}_B$:
$\res(\mathit{op}_A) \prec \inv(\mathit{op}_B)$.
The monitor fires $\eres(\mathit{op}_A)$ after $\res(\mathit{op}_A)$
and $\einv(\mathit{op}_B)$ before $\inv(\mathit{op}_B)$.
There is no constraint forcing
$\lp(\eres(\mathit{op}_A)) \prec \lp(\einv(\mathit{op}_B))$:
if $p_j$ fires $\einv(\mathit{op}_B)$ very early (long before
$\inv(\mathit{op}_B)$) and $p_i$ fires $\eres(\mathit{op}_A)$
late (just after $\res(\mathit{op}_A)$), then
$\lp(\einv(\mathit{op}_B)) \prec \lp(\eres(\mathit{op}_A))$
even though $\mathit{op}_A <_H \mathit{op}_B$.
The monitor observes $\einv(\mathit{op}_B)$ before
$\eres(\mathit{op}_A)$ and reports
$\precedes(A,B) = \mathtt{false}$, missing a genuine precedence.
\qed
\end{proof}

\begin{lemma}[Internal placement is not sound]
\label{lem:internal-not-sound}
A monitor with internal placement is complete for $\mathbf{COP}$
(Lemma~\ref{lem:placement}) but not sound: there exist executions
where $\mathit{op}_A \parallel \mathit{op}_B$ yet
$\precedes(A,B) = \mathtt{true}$.
\end{lemma}

\begin{proof}
Internal placement gives
$\inv(\mathit{op}) \prec \lp(\einv(\mathit{op}))$ and
$\lp(\eres(\mathit{op})) \prec \res(\mathit{op})$.

Consider two operations $\mathit{op}_A$ on $p_i$ and
$\mathit{op}_B$ on $p_j$ with $\mathit{op}_A \parallel
\mathit{op}_B$: their execution intervals overlap.
The monitor fires $\eres(\mathit{op}_A)$ before
$\res(\mathit{op}_A)$ and $\einv(\mathit{op}_B)$ after
$\inv(\mathit{op}_B)$.
Since the operations overlap and the processes are asynchronous,
it is possible that $\lp(\eres(\mathit{op}_A)) \prec
\lp(\einv(\mathit{op}_B))$ even though the operations are
concurrent: $p_i$ fires $\eres(\mathit{op}_A)$ early in its
interval while $p_j$ fires $\einv(\mathit{op}_B)$ late in its
interval.
The monitor observes $\eres(\mathit{op}_A)$ before
$\einv(\mathit{op}_B)$ and reports
$\precedes(A,B) = \mathtt{true}$, but
$\mathit{op}_A \parallel \mathit{op}_B$ so no real precedence
exists.
A spurious precedence is reported.
\qed
\end{proof}

\subsection{Impossibility of Strong Consistency}
\begin{theorem}[Impossibility of strong consistency]
\label{thm:impossibility}
No monitor can be strongly consistent for $\mathbf{COP}$, even
with an ideal log.
\end{theorem}
\begin{proof}
We prove that no monitor $\mathcal{M}$ can be strongly consistent
for $\mathbf{COP}$, even under an ideal log
(Definition~\ref{def:ideal-log}) and even if $\mathcal{M}$ uses
shared-memory objects of arbitrary power.

\paragraph{Setup.}
Fix any monitor $\mathcal{M}$ and any two operations
$\mathit{op}_A$ on process $p_i$ and $\mathit{op}_B$ on process
$p_j$, where $p_i \ne p_j$.
By Definition~\ref{def:monitor} and the WLOG argument above,
we may assume without loss of generality
that $\mathcal{M}$ emits exactly one interface event per boundary:
$\eres(\mathit{op}_A)$ associated with the response of
$\mathit{op}_A$, and $\einv(\mathit{op}_B)$ associated with the
invocation of $\mathit{op}_B$.
Under an ideal log, each interface event takes effect at a single
linearization point: $\lp(\eres(\mathit{op}_A))$ and
$\lp(\einv(\mathit{op}_B))$.

The true boundary $\res(\mathit{op}_A)$ is a transfer of control
between $p_i$ and its caller.
The interface event $\eres(\mathit{op}_A)$ is a shared-memory
step.
These are events of distinct categories and therefore have
distinct points in real time: $\res(\mathit{op}_A) \ne
\lp(\eres(\mathit{op}_A))$.

Let $a^-$ be the linearization point of the last interface event
of $\mathit{op}_A$ that precedes $\res(\mathit{op}_A)$ in program
order, and let $a^+ = \lp(\eres(\mathit{op}_A))$.
The true boundary $\res(\mathit{op}_A)$ lies strictly inside the
open interval $(a^-, a^+)$: it cannot equal $a^-$ because $a^-$
is a shared-memory step that precedes the response in program
order, and it cannot equal $a^+$ for the same reason.
This window $(a^-, a^+)$ is unavoidable: the only way to close it
would be to make $\res(\mathit{op}_A)$ and $\eres(\mathit{op}_A)$
the same event, which is impossible since one is a control
transfer and the other is a shared-memory write.

Define the window $(b^-, b^+)$ around $\inv(\mathit{op}_B)$
symmetrically, where $b^- = \lp(\einv(\mathit{op}_B))$ and $b^+$
is the linearization point of the first interface event of
$\mathit{op}_B$ that follows $\inv(\mathit{op}_B)$ in program
order.
By the same argument, $\inv(\mathit{op}_B)$ lies strictly inside
$(b^-, b^+)$ and this window cannot be closed.

The following diagram illustrates the two windows and their
overlap:

\[
\underbrace{%
  \overbrace{a^- \;\cdots\; \res(\mathit{op}_A) \;\cdots\; a^+}^{%
    \text{window } (a^-, a^+)}
  \qquad
  \overbrace{b^- \;\cdots\; \inv(\mathit{op}_B) \;\cdots\; b^+}^{%
    \text{window } (b^-, b^+)}
}_{\text{overlap } W \ne \emptyset \text{ by adversarial scheduling}}
\]

Because $p_i$ and $p_j$ are asynchronous --- no bound on relative
speeds --- an adversary can schedule their steps so that the two
windows overlap:
\[
  W = \bigl(\max(a^-, b^-),\;\min(a^+, b^+)\bigr) \ne \emptyset.
\]
Concretely: the adversary slows $p_j$ until $\einv(\mathit{op}_B)$
fires at the same time that $p_i$ is executing inside $(a^-, a^+)$,
and symmetrically delays $p_i$ so that $\eres(\mathit{op}_A)$ fires
while $p_j$ is inside $(b^-, b^+)$.
Since the two processes share no synchronization that could prevent
this interleaving, such a schedule always exists.

\paragraph{The indistinguishable executions.}
With $W \ne \emptyset$, we place $\res(\mathit{op}_A)$ and
$\inv(\mathit{op}_B)$ at two different positions within $W$:

\medskip
\begin{center}
\begin{tabular}{lll}
\toprule
& \textbf{Position of true boundaries} & \textbf{Relation} \\
\midrule
$E_1$ & $\res(\mathit{op}_A) \prec \inv(\mathit{op}_B)$ in $W$
      & $\mathit{op}_A <_H \mathit{op}_B$ \\
$E_2$ & $\inv(\mathit{op}_B) \prec \res(\mathit{op}_A)$ in $W$
      & $\mathit{op}_A \parallel \mathit{op}_B$ \\
\bottomrule
\end{tabular}
\end{center}
\medskip

Both executions are identical in every shared-memory step: the
interface events $\eres(\mathit{op}_A)$ and $\einv(\mathit{op}_B)$
fire at the same times $a^+$ and $b^-$ respectively, and all other
monitor steps are unchanged.
Since $\res(\mathit{op}_A)$ and $\inv(\mathit{op}_B)$ are not
shared-memory steps, moving them within $W$ does not alter any
linearization point, any register value, or any causal chain
observable by $\mathcal{M}$.
Under an ideal log, the log of $\mathcal{M}$ is therefore
identical in $E_1$ and $E_2$.

$\mathcal{M}$ must return a single answer for $\precedes(A,B)$
based on its log, and that answer is the same in $E_1$ and $E_2$.
But strong consistency requires:
\begin{itemize}
  \item In $E_1$: $\mathit{op}_A <_H \mathit{op}_B$, so
        completeness requires $\precedes(A,B) = \mathtt{true}$.
  \item In $E_2$: $\mathit{op}_A \parallel \mathit{op}_B$, so
        soundness requires $\precedes(A,B) = \mathtt{false}$.
\end{itemize}
No single answer satisfies both.
Therefore $\mathcal{M}$ cannot be strongly consistent for
$\mathbf{COP}$.

One might hope that a more powerful logging mechanism ---
one that somehow resolves the ambiguity inside $W$ --- could
escape the impossibility.
The ideal log already grants the strongest possible logging
assumption: perfect knowledge of the real-time order among all
interface events.
The impossibility does not stem from imprecision in the log.
It stems from the fact that $\res(\mathit{op}_A)$ and
$\inv(\mathit{op}_B)$ are not shared-memory steps and therefore
cannot appear in any log, no matter how ideal.
No logging mechanism, however powerful, can observe events that
are not shared-memory operations.
\qed
\end{proof}
\noindent
First identified in~\cite{CR23,CR26_JACM} as a structural
limitation of runtime verification of linearizability,
Theorem~\ref{thm:impossibility} shows this is in fact an
unconditional consequence of the observable boundary of any
concurrent object, independent of the verification algorithm and
of the shared-memory objects invoked.

\section{The Causal Monitor Object}
\label{sec:monitor}

The Causal Monitor is a concurrent object that instruments a
target algorithm $A$: processes invoke $\addinv(\mathit{id})$ at
the start of each operation and $\addres(\mathit{id})$ at its
end, creating the interface events of Section~\ref{sec:eint}, and
query the monitor via $\precedes$ to determine whether one
completed operation preceded another in real time. The sequential
specification below defines this ideal behavior: $\precedes(A,B)$
checks whether $A$'s response record appears before $B$'s
invocation record in the log, directly mirroring $<_H$; queries
require both response records to be present, since $<_H$ is only
defined for completed operations.

\begin{definition}[Causal Monitor---Sequential Specification]
\label{def:causal-monitor}
Let $\mathcal{I}$ be a universe of operation identifiers.
The abstract state is a sequence
$\mathcal{S} \in (\{I,R\} \times \mathcal{I})^*$,
initially empty; $\mathit{idx}(e,\mathcal{S})$ denotes the
position of $e$ in $\mathcal{S}$, or $\infty$ if absent.
\smallskip\noindent
\textbf{Modification.}
$\addinv(\mathit{id})$ requires $(I,\mathit{id}) \notin
\mathcal{S}$ and appends $(I,\mathit{id})$.
$\addres(\mathit{id})$ requires $(I,\mathit{id}) \in \mathcal{S}$
and $(R,\mathit{id}) \notin \mathcal{S}$, and appends
$(R,\mathit{id})$.
\smallskip\noindent
\textbf{Queries} (both require $(R,\mathit{id}_A),
(R,\mathit{id}_B) \in \mathcal{S}$; state unchanged).
$\precedes(\mathit{id}_A, \mathit{id}_B)$ returns $\mathtt{true}$
iff $\mathit{idx}((R,\mathit{id}_A),\mathcal{S}) < \mathit{idx}((I,\mathit{id}_B),\mathcal{S})$.
$\concurrent(\mathit{id}_A,\mathit{id}_B)$ returns $\mathtt{true}$
iff $\neg\,\precedes(\mathit{id}_A,\mathit{id}_B) \land
\neg\,\precedes(\mathit{id}_B,\mathit{id}_A)$.
\end{definition}

\section{Implementations}
\label{sec:implementations}

We present two implementations of the Causal Monitor, introduced
in~\cite{RC24} as practical monitoring tools but without establishing
their consistency properties or recognizing their shared $\mathbf{COP}$
guarantees --- the analysis below is a contribution of this work. The
first uses a shared counter and is \emph{linearizable}. The second
uses an iterative collect of process registers and, as
Section~\ref{sec:correctness} shows, satisfies \emph{none} of the
classical consistency conditions of Section~\ref{sec:model}---yet
remains \emph{sound} for $\mathbf{COP}$ (Definition~\ref{def:cop}):
soundness is strictly weaker than admitting a legal sequential witness
$S \in \mathit{Seq}(T)$. Under external placement, both resolve
$\mathbf{COP}$ identically and neither can be complete
(Theorem~\ref{thm:impossibility})---placement determines
$\mathbf{COP}$ guarantees; well-formedness as a concurrent object does
not.

\subsection{Counter Monitor}
\label{sec:counter}

The Counter Monitor assigns a monotonically increasing timestamp
$t \in \mathbb{N}$ to each boundary event upon invocation, via SWMR
registers $\Reg[1\dots N]$ (process $p_i$ appends to $\Reg[i]$) and a
shared counter $C$, in two modes (Algorithm~\ref{alg:counter}):
\textsc{FAInc} increments a centralized atomic counter via a single
\texttt{getAndIncrement}(), giving every event a unique, totally
ordered real-time timestamp; \textsc{Striped} distributes the counter
across $N$ per-process cells (as in Java's
\texttt{LongAdder}~\cite{lea-longadder}), incrementing the local cell
$C[i]$ and summing $\sum_{j=1}^N C[j]$ non-atomically, trading
occasional identical timestamps for eliminating cache-line bouncing at
high thread counts.

Instrumenting a boundary calls $\addinv(\mathit{id})$ or
$\addres(\mathit{id})$, each stamping $t \gets \textsc{Stamp}()$ and
appending $(I,\mathit{id},t)$ or $(R,\mathit{id},t)$ to $\Reg[i]$.
$\precedes(\mathit{id}_A,\mathit{id}_B)$ compares the timestamp stored
with $A$'s response to the one stored with $B$'s invocation, returning
\texttt{true} iff $t_A < t_B$.

\begin{algorithm}[htb]
\small
\caption{Counter-based Causal Monitor}
\label{alg:counter}
\begin{algorithmic}[1]
\Statex \textbf{Shared:} $\Reg[1\dots N]$ (SWMR, append-only);\quad counter $C$
\Statex \textbf{Parameter:} $\mathit{Mode} \in \{\textsc{FAInc},\,\textsc{Striped}\}$
\Function{Stamp}{}
  \If{$\mathit{Mode} = \textsc{FAInc}$}
    \State \Return $C.\texttt{getAndIncrement}()$
        \Comment{single atomic step}
  \Else
    \State $C.\texttt{increment}()$
    \State \Return $C.\texttt{sum}()$
        \Comment{increment local cell, non-atomically read all cells}
  \EndIf
\EndFunction
\Procedure{Add\_Inv}{$\mathit{id}$}
  \State $t \gets \textsc{Stamp}()$;\quad $\Reg[i] \gets \Reg[i]\cdot\langle(I,\mathit{id},t)\rangle$
\EndProcedure
\Procedure{Add\_Res}{$\mathit{id}$}
  \State $t \gets \textsc{Stamp}()$;\quad $\Reg[i] \gets \Reg[i]\cdot\langle(R,\mathit{id},t)\rangle$
\EndProcedure
\Function{Precedes}{$\mathit{id}_A,\mathit{id}_B$}
  \State $t_A \gets$ timestamp stored with $(R,\mathit{id}_A)$
  \State $t_B \gets$ timestamp stored with $(I,\mathit{id}_B)$
  \State \Return $t_A < t_B$
\EndFunction
\end{algorithmic}
\end{algorithm}

\paragraph{Implementation of \textsc{Striped}.}
\textsc{Striped} is implemented as an array of $N$ per-process atomic
cells $C[1\dots N]$, where process $p_i$ only increments $C[i]$; the
\textsc{Stamp} function for $p_i$ executes (1) $C[i].\texttt{incrementAndGet}()$
(atomic write to its own cell), then (2) \textbf{return}
$\sum_{j=1}^{N} C[j].\texttt{get}()$ (a non-atomic traversal of all
cells). In Java, this uses an \texttt{AtomicLongArray} with memory
padding between cells to prevent false sharing on cache lines: each
cell occupies a full cache line (typically 64 bytes), so increments by
different processes do not invalidate each other's cache entries. The
key invariant is that the increment always precedes the sum within
the same \textsc{Stamp} call, which ensures that the value returned
by $\texttt{sum}()$ is at least as large as the calling process's own
contribution, and that the global sum is strictly monotone across
non-overlapping calls.

\begin{proposition}[Counter Monitor is linearizable]
\label{prop:counter-lin}
Both \textsc{FAInc} and \textsc{Striped} produce linearizable
executions of the Causal Monitor.
\end{proposition}

\begin{proof}
\emph{\textsc{FAInc}.} The linearization point of each $\addinv$ or
$\addres$ call is its atomic \texttt{getAndIncrement}; timestamps
increase strictly in real time, so the timestamp order satisfies the
specification.

\emph{\textsc{Striped}.} We show every execution admits a
linearization: a legal sequential history in which each operation's
linearization point $\lp(\cdot)$ falls within its execution interval.

Consider any \textsc{Stamp} call by process $p_i$ that returns value
$t$. The call first increments $C[i]$, then reads the sum
$\sum_{j=1}^{N} C[j]$. Because each cell $C[j]$ is monotonically
non-decreasing and the increment of $C[i]$ occurs before the sum
traversal within the same call, the sum at the start of the traversal
is at least $t$ (since $C[i]$ already reflects the increment) and at
most the sum at the end of the traversal. More precisely: let
$s_{\mathit{start}}$ be the true global sum at the moment the
traversal begins and $s_{\mathit{end}}$ be the true global sum at the
moment it ends. The returned value $t$ satisfies $s_{\mathit{start}}
\le t \le s_{\mathit{end}}$, since the non-atomic traversal reads each
cell at some point between its start and end. Because each cell
increments by exactly one at a time, the global sum increases in unit
steps. Therefore, there exists a physical instant $\lp$ during the
traversal at which the true global sum equals exactly $t$. This
instant $\lp$ lies within the execution interval of the
\textsc{Stamp} call and serves as its linearization point.

Let $s_1$ and $s_2$ be two \textsc{Stamp} calls where $s_1$ completes
before $s_2$ begins. Since every \textsc{Stamp} increments before
summing, the increment of $s_1$ is complete before $s_2$ begins. When
$s_2$ executes its sum traversal, it reads every cell after $s_1$'s
increment has taken effect, so it observes a global sum strictly
greater than $s_1$'s contribution. In particular, $\ts(s_2) \ge
\ts(s_1) + 1 > \ts(s_1)$: non-overlapping calls receive strictly
increasing timestamps.

Suppose two calls $s_1$ and $s_2$ return the same value $t$. By the
argument above, non-overlapping calls receive strictly different
timestamps, so $s_1$ and $s_2$ must overlap in real time. The tie may
be broken arbitrarily in the sequential history without violating the
specification, since the Causal Monitor's sequential specification
imposes no order on concurrent operations.

For a query $\precedes(A,B)$ returning \texttt{true}, we have $t_A <
t_B$, where $t_A$ is the timestamp of $\addres(A)$ and $t_B$ is the
timestamp of $\addinv(B)$. By monotonicity, $t_A < t_B$ implies
$\lp(\addres(A)) \prec \lp(\addinv(B))$, so $\addres(A)$ precedes
$\addinv(B)$ in the linearization. This matches the sequential
specification, which returns \texttt{true} exactly when $A$'s
response record precedes $B$'s invocation record in the state.
\qed
\end{proof}

\noindent
\emph{Note on false sharing and padding.} The \texttt{AtomicLongArray}
with padding is structurally similar to Java's
\texttt{LongAdder}~\cite{lea-longadder}, but with a critical
difference: \texttt{LongAdder} supports both increment and decrement
and is designed for accumulation, not for monotone counting.
\textsc{Striped} uses only increment operations, which is what
guarantees the monotonicity argument above; allowing decrements would
break the intermediate value argument and invalidate the
linearizability proof. The padding prevents cache-line false sharing
at the hardware level but does not alter the logical monotonicity or
the linearizability of the algorithm in any way.

\noindent
\textsc{FAInc} was introduced in~\cite{RC24} grouped with Collect
(Algorithm~\ref{alg:collect}) as a variant of the same approach;
\textsc{Striped} is introduced here, following~\cite{lea-longadder}.
Section~\ref{sec:correctness} shows this grouping was imprecise: both
counters are linearizable, whereas Collect admits no sequential
witness at all --- yet all three are sound for $\mathbf{COP}$ under
external placement, without being equivalent as concurrent objects.

\subsection{Collect Monitor}
\label{sec:collect}

The Collect Monitor maintains explicit causal history views instead
of a numeric counter, via append-only SWMR registers $\Reg[1\dots N]$
where each process $p_i$ logs invocation and response events
(Algorithm~\ref{alg:collect}). $\addinv(\mathit{id})$ appends
$(I,\mathit{id})$ to $\Reg[i]$ without inspecting other registers.
$\addres(\mathit{id})$ performs an unsynchronized, iterative
\emph{collect} $v \gets \bigcup_{j=1}^N \Reg[j]$ across all registers,
then appends $(R,\mathit{id},v)$, embedding the snapshot view
alongside the response. $\precedes(\mathit{id}_A,\mathit{id}_B)$
retrieves the views $v_A, v_B$ captured at $A$'s and $B$'s response
and returns \texttt{true} iff all three hold: $(R,\mathit{id}_A) \in v_B$
($A$'s response was visible when $B$ collected), $(I,\mathit{id}_B) \notin v_A$
($B$'s invocation was not yet visible when $A$ collected), and
$v_A \subseteq v_B$ (everything $A$ observed, $B$ also observed,
restoring transitivity across non-overlapping collects).

\begin{algorithm}[H]
\small
\caption{Collect-based Causal Monitor}
\label{alg:collect}
\begin{algorithmic}[1]
\Statex \textbf{Shared:} $\Reg[1\dots N]$ (SWMR, append-only sequences of events)
\Procedure{Add\_Inv}{$\mathit{id}$}
  \State $\Reg[i] \gets \Reg[i]\cdot\langle(I,\mathit{id})\rangle$
\EndProcedure
\Procedure{Add\_Res}{$\mathit{id}$}
  \State $v \gets \bigcup_{j=1}^{N}\Reg[j]$
      \Comment{iterative, non-atomic read across all process registers}
  \State $\Reg[i] \gets \Reg[i]\cdot\langle(R,\mathit{id},v)\rangle$
\EndProcedure
\Function{Precedes}{$\mathit{id}_A,\mathit{id}_B$}
  \State $v_A \gets$ view stored with $(R,\mathit{id}_A)$
  \State $v_B \gets$ view stored with $(R,\mathit{id}_B)$
  \State \Return $[(R,\mathit{id}_A) \in v_B] \;\land\; [(I,\mathit{id}_B) \notin v_A] \;\land\; [v_A \subseteq v_B]$
\EndFunction
\end{algorithmic}
\end{algorithm}

\noindent
The view-containment conjunct's role is purely structural, restoring
transitivity of $\precedes$; it does not affect soundness
(Section~\ref{sec:soundness}).

\section{Correctness}
\label{sec:correctness}

This section establishes two facts whose relationship is the main
technical contribution of the paper.
First, \emph{both monitors are sound for $\mathbf{COP}$} under external
placement (Theorems~\ref{thm:counter-sound}
and~\ref{thm:collect-sound}): whenever a monitor reports
$\precedes(A,B) = \mathtt{true}$, the operation $\mathit{op}_A$
genuinely completed before $\mathit{op}_B$ began in real time,
i.e.\ $\mathit{op}_A <_H \mathit{op}_B$.
Second, \emph{as concurrent objects, the two monitors could not be more
different}: the Counter Monitor is linearizable
(Proposition~\ref{prop:counter-lin}), while the Collect Monitor admits
\emph{no} legal sequential witness at all
(Theorem~\ref{thm:collect-no-witness})---it satisfies none of
linearizability, sequential consistency, or quiescent consistency, all
three of which require the existence of some $S \in \mathit{Seq}(T)$
equivalent to a completion of the produced history
(conditions~(L1),~(SC1),~(QC1) of Section~\ref{sec:model}).
 
The juxtaposition of these two facts is the point: $\mathbf{COP}$-soundness
(Definition~\ref{def:cop}) never requires the existence of a global
sequential witness. Its proof (Theorem~\ref{thm:collect-sound}) is
established \emph{per query}, directly from the physical ordering of
interface events, and never invokes the existence of any $S$.
The Collect Monitor shows that this is not a vacuous distinction: a
mechanism can answer every $\mathbf{COP}$ query correctly in isolation
while the collection of all its answers, taken together, is not
projectable onto any single legal sequential history. Soundness and
membership in the classical consistency hierarchy are, in general,
independent properties of a monitor.

 \subsection{Soundness}
\label{sec:soundness}

\begin{theorem}[Counter Monitor is sound for $\mathbf{COP}$]
\label{thm:counter-sound}
Under external placement, Algorithm~\ref{alg:counter} is sound:
$\precedes(A,B) = \mathtt{true} \Rightarrow \mathit{op}_A <_H \mathit{op}_B$.
\end{theorem}

\begin{proof}
$\precedes(A,B) = \mathtt{true}$ means $t_A < t_B$, where $t_A$
is the timestamp stamped during $\addres(A)$ and $t_B$ is the
timestamp stamped during $\addinv(B)$.
In both \textsc{FAInc} and \textsc{Striped}, timestamps are
monotone in real time: a strictly smaller timestamp means the
corresponding \textsc{Stamp} call completed before the other
began.
Let $\sigma_A$ be the linearization point of the \textsc{Stamp}
call inside $\addres(A)$, and $\sigma_B$ the linearization point
of the \textsc{Stamp} call inside $\addinv(B)$.
Then $t_A < t_B$ implies $\sigma_A \prec \sigma_B$.

External placement gives
$\res(\mathit{op}_A) \prec \sigma_A$,
since $\addres(A)$ is invoked after $\mathit{op}_A$ returns and
$\sigma_A$ lies within $\addres(A)$'s execution interval.
Similarly, $\sigma_B \prec \inv(\mathit{op}_B)$,
since $\addinv(B)$ is invoked before $\mathit{op}_B$ starts and
$\sigma_B$ lies within $\addinv(B)$'s execution interval.
Chaining, $\res(\mathit{op}_A) \prec \sigma_A \prec \sigma_B \prec
\inv(\mathit{op}_B)$, so $\res(\mathit{op}_A) \prec \inv(\mathit{op}_B)$,
i.e.\ $\mathit{op}_A <_H \mathit{op}_B$.
\qed
\end{proof}

\begin{theorem}[Collect Monitor is sound for $\mathbf{COP}$]
\label{thm:collect-sound}
Under external placement, Algorithm~\ref{alg:collect} is sound:
$\precedes(A,B) = \mathtt{true} \Rightarrow \mathit{op}_A <_H \mathit{op}_B$.
\end{theorem}

\begin{proof}
Suppose $\precedes(A,B) = \mathtt{true}$.
The predicate has three conjuncts; we use only the second:
$(I,\mathit{id}_B) \notin v_A$.
This means that when $A$'s collect read $\Reg[\mathrm{proc}(B)]$,
the write of $(I,\mathit{id}_B)$ had not yet occurred.
Let $\rho$ be the physical time of that read and $w_I$ the
physical time of the write of $(I,\mathit{id}_B)$;
then $\rho \prec w_I$.

External placement gives two facts.
First, $\addres(A)$ is invoked after $\mathit{op}_A$ returns,
so $\res(\mathit{op}_A) \prec \inv(\addres(A))$; since $\rho$
occurs during $\addres(A)$'s collect,
$\res(\mathit{op}_A) \prec \rho$.
Second, $\addinv(B)$ is invoked before $\mathit{op}_B$ starts,
so the write $w_I$ of $(I,\mathit{id}_B)$, which occurs inside
$\addinv(B)$, satisfies $w_I \prec \inv(\mathit{op}_B)$.
Chaining, $\res(\mathit{op}_A) \prec \rho \prec w_I \prec
\inv(\mathit{op}_B)$, so $\res(\mathit{op}_A) \prec \inv(\mathit{op}_B)$,
i.e.\ $\mathit{op}_A <_H \mathit{op}_B$.
\qed
\end{proof}

\noindent
This proof uses only the conjunct $(I,\mathit{id}_B) \notin v_A$,
and never references the existence of any $S \in \mathit{Seq}(T)$
witnessing the execution as a whole: the argument is entirely local
to the pair $(A,B)$.
The first conjunct $(R,\mathit{id}_A) \in v_B$ and the
view-containment conjunct $v_A \subseteq v_B$ play no role in
soundness.
Their purpose is structural: the first contributes to detecting
precedence, and the second restores transitivity of $\precedes$,
as established in Lemma~\ref{lem:partial-order}.
Soundness holds independently of whether $\precedes$ is transitive,
and, as Section~\ref{sec:collect-consistency} shows, independently
of whether the monitor is a well-formed concurrent object at all.

\subsection{The Collect Monitor admits no sequential witness}
\label{sec:collect-consistency}
 
\begin{lemma}[$\precedes$ is a strict partial order under Collect]
\label{lem:partial-order}
With the view-containment conjunct, $\precedes$ is irreflexive,
antisymmetric, and transitive.
\end{lemma}
 
\begin{proof}
We prove that $\precedes$ under the Collect Monitor
(Algorithm~\ref{alg:collect}) is a strict partial order, assuming
the view-containment conjunct $v_A \subseteq v_B$ is included in
the predicate. Recall that $\precedes(\mathit{id}_A, \mathit{id}_B)$
returns $\mathtt{true}$ iff all three of the following hold:
(C1) $(R,\mathit{id}_A) \in v_B$; (C2) $(I,\mathit{id}_B) \notin
v_A$; (C3) $v_A \subseteq v_B$.
\begin{itemize}
 \item Irreflexivity: $\precedes(A,A) = \mathtt{false}$.
The collect of $\addres(A)$ runs to completion before the write
of $(R,A)$ to $\Reg[\mathrm{proc}(A)]$ occurs --- the collect
is part of $\addres(A)$ and the write happens at the end of it.
Therefore $(R,A) \notin v_A$, which violates~(C1).
Hence $\precedes(A,A) = \mathtt{false}$.

\item Antisymmetry: $\precedes(A,B) \Rightarrow \neg\precedes(B,A)$.
  
Suppose $\precedes(A,B) = \mathtt{true}$.
By~(C1), $(R,A) \in v_B$.
Since registers are append-only and process $\mathrm{proc}(A)$
appends $(I,A)$ strictly before $(R,A)$ on $\Reg[\mathrm{proc}(A)]$,
and $v_B$ is a prefix-closed snapshot of the registers, it follows
that $(I,A) \in v_B$ as well.
Now consider $\precedes(B,A)$: it would require $(I,A) \notin v_B$
by~(C2) with the roles of $A$ and $B$ swapped.
But we just showed $(I,A) \in v_B$, so~(C2) fails and
$\precedes(B,A) = \mathtt{false}$.

\item Transitivity: $\precedes(A,B) \land \precedes(B,C)
\Rightarrow \precedes(A,C)$.
Suppose $\precedes(A,B) = \mathtt{true}$ and
$\precedes(B,C) = \mathtt{true}$.
We verify the three conjuncts of $\precedes(A,C)$.
 
(C1) for $\precedes(A,C)$: $(R,A) \in v_C$.
From $\precedes(A,B)$, conjunct~(C1) gives $(R,A) \in v_B$.
From $\precedes(B,C)$, conjunct~(C3) gives $v_B \subseteq v_C$.
Therefore $(R,A) \in v_C$.
 
(C2) for $\precedes(A,C)$: $(I,C) \notin v_A$.
From $\precedes(B,C)$, conjunct~(C2) gives $(I,C) \notin v_B$.
From $\precedes(A,B)$, conjunct~(C3) gives $v_A \subseteq v_B$.
Since $(I,C) \notin v_B$ and $v_A \subseteq v_B$, we have
$(I,C) \notin v_A$.
 
(C3) for $\precedes(A,C)$: $v_A \subseteq v_C$.
From $\precedes(A,B)$, conjunct~(C3) gives $v_A \subseteq v_B$.
From $\precedes(B,C)$, conjunct~(C3) gives $v_B \subseteq v_C$.
By transitivity of subset inclusion, $v_A \subseteq v_C$.
 
All three conjuncts hold, so $\precedes(A,C) = \mathtt{true}$.
  \qed
\end{itemize}

\end{proof}
 
\noindent
Note that view-containment~(C3) is essential for transitivity:
without it, (C1) and (C2) alone are not transitive.
Specifically, (C2) for $\precedes(A,C)$ requires
$(I,C) \notin v_A$, but knowing $(I,C) \notin v_B$ and
$v_A \subseteq v_B$ is what allows us to conclude this.
Without~(C3), $v_A \subseteq v_B$ would not be available and
the derivation would fail.
 
\noindent
Lemma~\ref{lem:partial-order} shows that $\precedes$ is well-behaved
as a \emph{pairwise} relation. The next theorem shows that this
pairwise good behavior does not extend to the existence of a single
sequential history consistent with \emph{all} pairs simultaneously---the additional requirement common to (L1), (SC1), and (QC1).
 
\begin{theorem}[Collect Monitor admits no sequential witness]
\label{thm:collect-no-witness}
There is an execution of Algorithm~\ref{alg:collect} for which no
sequential history $S$ is equivalent to a completion of the produced
history (Definition~\ref{def:equiv}). Consequently, Algorithm~\ref{alg:collect}
satisfies none of linearizability, sequential consistency, or
quiescent consistency (Definitions~\ref{def:lin}--\ref{def:qc}), since
all three require such an $S$ to exist as condition~(L1), (SC1), or
(QC1) respectively.
\end{theorem}
 
\begin{proof}
We construct a concrete execution of four operations $A, B, C, D$
on distinct processes $p_A, p_B, p_C, p_D$ whose verdicts under
Algorithm~\ref{alg:collect} admit no sequential history.
 
\paragraph{Write times.}
The interface events are written to shared memory at the following
physical times:
\[
w(I_A){=}1,\quad
w(I_C){=}2,\quad
w(I_B){=}3,\quad
w(I_D){=}4,\quad
w(R_C){=}5,\quad
w(R_A){=}6.
\]
\paragraph{Collect intervals.}
Each process performs its collect for $\addres$ by reading the
registers at the following times:
\begin{itemize}
  \item $A$ collects during $[2.0, 2.5)$: reads all registers
        while only $I_A$ and $I_C$ have been written.
  \item $C$ collects during $[2.5, 3.0)$: reads all registers
        while only $I_A$ and $I_C$ have been written.
  \item $B$ reads $\Reg[\mathrm{proc}(C)]$ at $t{=}4$
        (before $w(R_C){=}5$) and $\Reg[\mathrm{proc}(A)]$
        at $t{=}7$ (after $w(R_A){=}6$).
  \item $D$ reads $\Reg[\mathrm{proc}(A)]$ at $t{=}5$
        (before $w(R_A){=}6$) and $\Reg[\mathrm{proc}(C)]$
        at $t{=}6$ (after $w(R_C){=}5$).
\end{itemize}
This is the key crossing: $B$ sees $A$'s response but misses
$C$'s, while $D$ sees $C$'s response but misses $A$'s.
\paragraph{Resulting views.}
\[
\begin{aligned}
v_A &= \{I_A, I_C\}, \\
v_B &= \{I_A, R_A, I_B, I_C, I_D\}, \\
v_C &= \{I_A, I_C\}, \\
v_D &= \{I_A, I_C, R_C, I_D, I_B\}.
\end{aligned}
\]
\paragraph{Evaluating the predicate.}
We check $\precedes(X,Y)$
$= [(R,X) \in v_Y] \land [(I,Y) \notin v_X] \land [v_X \subseteq v_Y]$
for each relevant pair:
\begin{itemize}
  \item \emph{$\precedes(A,B) = \mathtt{true}$}, since all three
        conjuncts hold: $(R,A) \in v_B$, $(I,B) \notin v_A$, and
        $v_A = \{I_A,I_C\} \subseteq v_B$.
  \item \emph{$\precedes(C,D) = \mathtt{true}$}, since all three
        conjuncts hold: $(R,C) \in v_D$, $(I,D) \notin v_C$, and
        $v_C = \{I_A,I_C\} \subseteq v_D$.
  \item \emph{$\precedes(C,B) = \mathtt{false}$}, since
        conjunct~(C1) fails: $(R,C) \notin v_B$ because $B$ read
        $\Reg[\mathrm{proc}(C)]$ at $t{=}4$, before $w(R_C){=}5$.
  \item \emph{$\precedes(A,D) = \mathtt{false}$}, since
        conjunct~(C1) fails: $(R,A) \notin v_D$ because $D$ read
        $\Reg[\mathrm{proc}(A)]$ at $t{=}5$, before $w(R_A){=}6$.
\end{itemize}
\paragraph{No sequential witness exists.}
Suppose, towards a contradiction, that some sequential history $S$
is equivalent to a completion of this execution (condition (L1),
(SC1), or (QC1) of Section~\ref{sec:model}; all three share this same
requirement). By the sequential specification of the Causal Monitor
(Definition~\ref{def:causal-monitor}), $\precedes(X,Y) = \mathtt{true}$
in $S$ holds iff $(R,X)$ precedes $(I,Y)$ at an earlier position of
$S$, and since $S$ is total, $\precedes(X,Y) = \mathtt{false}$ forces
the reverse strict order. The four verdicts above therefore impose:
\[
\precedes(A,B) = \mathtt{true}
\;\Rightarrow\; R_A \prec_S I_B,
\qquad
\precedes(C,D) = \mathtt{true}
\;\Rightarrow\; R_C \prec_S I_D,
\]
\[
\precedes(C,B) = \mathtt{false}
\;\Rightarrow\; I_B \prec_S R_C,
\qquad
\precedes(A,D) = \mathtt{false}
\;\Rightarrow\; I_D \prec_S R_A.
\]
Chaining these four constraints:
\[
R_A \prec_S I_B \prec_S R_C \prec_S I_D \prec_S R_A,
\]
a cycle, contradicting that $\prec_S$ is a strict order on the
positions of $S$. No sequential history is consistent with all four
verdicts simultaneously, so no $S$ as required by (L1), (SC1), or
(QC1) exists for this execution.
 
Since linearizability, sequential consistency, and quiescent
consistency (Definitions~\ref{def:lin}--\ref{def:qc}) all require the
existence of such an $S$ as part of their definition, the Collect
Monitor satisfies \emph{none} of the three for this execution ---
not only failing to be linearizable, but failing every classical
consistency condition that presupposes a global sequential witness.
The obstruction is combinatorial rather than representational: treating
each operation as an atomic interval (rather than as two separate
points $\mathit{inv}$ and $\mathit{res}$ in $S$) does not open any
additional freedom, since (L1)/(SC1)/(QC1) already require $S$ to be
\emph{sequential} (Section~\ref{sec:model}), meaning no operation's
invocation and response are separated in $S$ by another operation's
events --- operations are already indivisible blocks within $S$, and
the four chained constraints above reduce to the same cycle under
either representation.
\qed
\end{proof}


\section{Empirical Evaluation}
\label{sec:evaluation}

Section~\ref{sec:correctness} shows that the Counter Monitor
(\textsc{FAInc}, \textsc{Striped}) and the Collect Monitor solve
$\mathbf{COP}$ identically under external placement, despite the
Counter Monitor being well-formed linearizable while Collect admits no
sequential witness (Theorem~\ref{thm:collect-no-witness}). We evaluate
this empirically via \emph{runtime verification (RV) of
linearizability}~\cite{RPC26}, where a monitor decides whether an
observed history is linearizable~\cite{CR26_JACM}: \emph{RV-complete}
if it never rejects a genuinely linearizable execution,
\emph{RV-sound} if it never accepts a non-linearizable one. By
Lemma~\ref{lem:placement}, external placement guarantees
$\mathbf{COP}$ soundness, so all three monitors are RV-complete by
construction; RV-soundness is what the evaluation probes.

\textsc{Striped}, introduced in this work, distributes increments
across per-process cells to eliminate \textsc{FAInc}'s cache-line
contention. \textsc{FAInc} vs.\ Collect was already compared
in~\cite{RC24,RPC26}; we add \textsc{Striped}, asking whether it
matches \textsc{FAInc}'s detection rate despite its non-atomic sum
(\textbf{RQ1: verdict consistency}) and eliminates \textsc{FAInc}'s
contention while remaining competitive with Collect (\textbf{RQ2:
scalability}). We implemented all three in Java on a 64-core AMD NUMA
architecture, following Georges et al.'s
methodology~\cite{georges2007java} (5 warmup / 10 measured rounds, 100
runs per cell). RQ1 is the bug detection rate on a seeded
\texttt{NonLinearizableQueue} across thread counts; RQ2 is mean
execution time and maximum latency at 10{,}000 operations per thread
(Figure~\ref{fig:results}, left and right).

\begin{figure*}[t]
  \centering
  \includegraphics[width=0.30\textwidth]{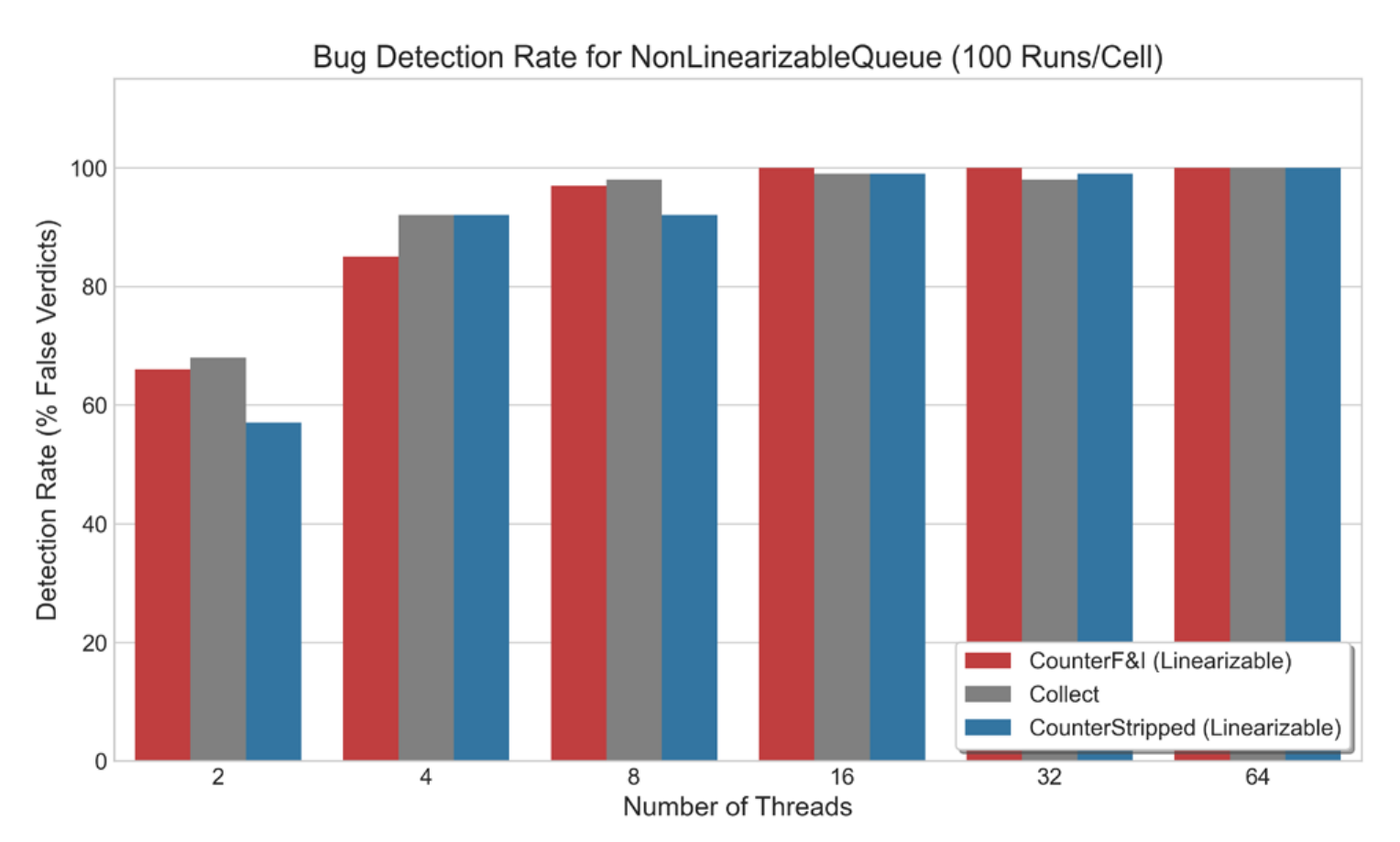}\hfill
  \includegraphics[width=0.67\textwidth]{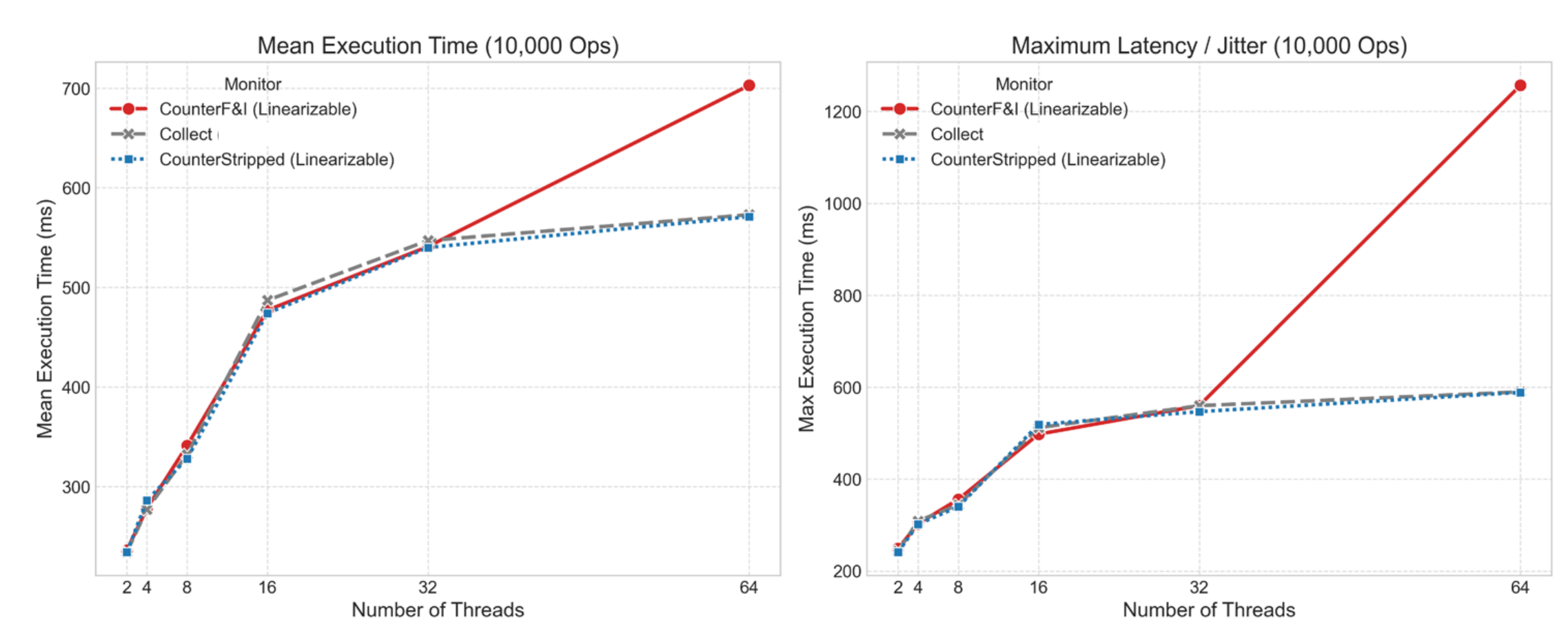}
  \caption{Left: bug detection rate on \texttt{NonLinearizableQueue} (100 runs per
  cell). Right: mean execution time and maximum latency at 10,000 operations
  per thread.}
  \label{fig:results}
\end{figure*}

\subsection{RQ1: Verdict Consistency}
\label{sec:rq1}

Figure~\ref{fig:results} (left) shows the detection rate on a
\texttt{NonLinearizableQueue} across thread counts.
\textsc{FAInc} reaches 100\% detection consistently from low thread
counts.
\textsc{Striped} starts slightly below \textsc{FAInc} at 2 threads
but converges to 100\% at 16 threads and remains there.
The Collect Monitor lags behind both counters at low thread counts.

The minor gap between \textsc{Striped} and \textsc{FAInc} at low
concurrency is a direct consequence of the non-atomic
\texttt{sum()}: the traversal takes physical time, slightly widening
the observed interval of each operation and granting the verifier
a marginally larger concurrent window in which violations can be
hidden.
This is the same mechanism that limits Collect at low concurrency,
though Collect's effect is more pronounced because its widening is
proportional to the number of registers read rather than counter
cells.
As thread contention grows, the widening becomes negligible relative
to the actual concurrency in the execution, and \textsc{Striped}
converges to \textsc{FAInc}.

These results confirm that, in practice, the three monitors are
equivalent for the purpose of RV: all are RV-complete, and all
achieve equivalent detection rates at realistic thread counts. This
holds despite Collect admitting no sequential witness at all
(Theorem~\ref{thm:collect-no-witness}): $\mathbf{COP}$-soundness,
established per query, is unaffected by that gap, and so is the
detection behavior it enables.

\subsection{RQ2: Scalability}
\label{sec:rq2}

Figure~\ref{fig:results} (right) shows mean execution time and
maximum latency at 10,000 operations per thread.
Up to 16 threads, all three monitors perform similarly: modern cache
coherency protocols mask the contention on the centralized
\textsc{FAInc} counter.
Beyond 16 threads the difference becomes visible, and at 64 threads
\textsc{FAInc} suffers a severe degradation: mean time reaches
700\,ms and maximum latency exceeds 1,200\,ms, evidencing
cache-line bouncing on the centralized atomic counter under maximum
hardware concurrency.

\textsc{Striped} and Collect remain stable across all thread counts,
with mean times near 575\,ms and maximum latency near 580\,ms at
64 threads.
\textsc{Striped} is marginally faster than Collect at low thread
counts, where the per-register traversal of Collect adds a small
but measurable overhead; at high thread counts the two are
effectively indistinguishable.

These results answer RQ2 affirmatively: \textsc{Striped} eliminates
the scalability bottleneck of \textsc{FAInc} while matching Collect
in throughput.
Combined with the verdict results of RQ1, the conclusion is that
\textsc{Striped} is a strictly better choice than \textsc{FAInc}
as a linearizable implementation of the Causal Monitor: it provides
the same $\mathbf{COP}$ guarantees, the same detection behavior at
realistic concurrency, and scales to high thread counts without
contention collapse.

The comparison with Collect, however, is not one of dominance.
Collect was shown in~\cite{RC24,RPC26} to outperform \textsc{FAInc}
precisely because it avoids centralized synchronization.
\textsc{Striped} now shows that a \emph{linearizable} counter can
match that performance. The gap between the Counter Monitor's
linearizability and Collect admitting no sequential witness at all
--- made precise by Theorem~\ref{thm:collect-no-witness} --- does not
translate into a performance gap at the scale tested, even though
Collect's correctness profile is strictly weaker than any classical
consistency condition, not merely a lower point within their
hierarchy. Whether the two diverge under more adversarial workloads
is left for future work.

Both RQ1 and RQ2 confirm the predictions of
Section~\ref{sec:correctness}: the three monitors are observationally
equivalent for $\mathbf{COP}$, and \textsc{Striped} matches Collect's
scalability while retaining linearizability.

\section{Discussion and Further Work}
\label{sec:discussion}

$\mathbf{COP}$ applies wherever a system must track operation order
from an external vantage point---e.g.\ causally consistent
replication, where \textit{GentleRain} and
\textit{Cure}~\cite{GentleRain,Cure} track causal timestamps over
client-transaction boundaries~\cite{Lamport78}, subject to the same
soundness/completeness trade-off identified here; and distributed
tracing, where Dapper~\cite{Dapper10} and X-Trace~\cite{XTrace07}
instrument only service-call entry/exit points, facing the same
observable-boundary problem and explaining why some tracers report
spurious dependencies while others miss genuine ones.

Causality was first characterized for general asynchronous systems
in~\cite{RS96}. Under Byzantine failures, causality detection in
\emph{message-passing} systems is unsolvable asynchronously for
classical primitives~\cite{MK22}---a different model from ours
($\mathbf{COP}$ uses shared-memory happened-before,
Definition~\ref{def:lamport}); whether an analogous impossibility
holds for $\mathbf{COP}$ under Byzantine shared-memory processes is
open.

Targeting $\hbL$ rather than $<_H$ over the full execution $\Eall$
makes strong consistency achievable---vector
clocks~\cite{Fidge88,Mattern88,RS96} solve it exactly, since $\hbL$
and $\hbL[\Eall]$ coincide. Shrinking the observable set to $\Eobs$ is
what makes $\hbL[\Eobs]$ and $<_H$ diverge (Section~\ref{sec:cop}),
and strong consistency unachievable; characterizing which
observable-set/target-relation pairs admit strong consistency is open.

Collect admits no sequential witness
(Theorem~\ref{thm:collect-no-witness}) yet is $\mathbf{COP}$-sound
(Theorem~\ref{thm:collect-sound}) via a purely per-query argument,
suggesting $\mathbf{COP}$-soundness may be achievable by mechanisms
weaker than well-formed objects~\cite{HW90}.
Whether Collect satisfies a genuinely weaker condition---e.g. Interval-Linearizability~\cite{CRR23}, with a witness over
intervals rather than points---and whether a cheaper
$\mathbf{COP}$-sound algorithm exists, are left as future work.

\bibliography{biblio}
\bibliographystyle{splncs04}

\end{document}